\documentclass[conference, letterpaper, romanappendices]{IEEEtran}
\usepackage{graphics}
\usepackage{cite}
\usepackage[pdftex]{graphicx}
\usepackage{epstopdf}
\usepackage{epsfig}
\usepackage{latexsym}
\usepackage{amsfonts}
\usepackage{calc}
\usepackage{url}
\usepackage{enumerate}
\usepackage{color}
\usepackage[tbtags]{amsmath}
\usepackage{amssymb}
\usepackage{upref}
\usepackage{dsfont}
\usepackage{multirow}
\usepackage{booktabs}
\usepackage{bigstrut}
\usepackage{rotating}
\usepackage{bbm}
\usepackage{breqn}
\usepackage{nth}
\usepackage{bm}

\newtheorem{theorem}{Theorem}

\newtheorem{lemma}{Lemma}

\newtheorem{definition}{Definition}
\newtheorem{proof}{Proof}
\newtheorem{conjecture}{Conjecture}

\begin{document}

\title{An Algebraic-Combinatorial Proof Technique for the GM-MDS Conjecture}

\author{Anoosheh Heidarzadeh and Alex Sprintson\\ Texas A\&M University, College Station, TX 77843 USA}

\maketitle
\thispagestyle{empty}  





\begin{abstract}
This paper considers the problem of designing maximum distance separable (MDS) codes over small fields with constraints on the support of their generator matrices. For any given $m\times n$ binary matrix $M$, the \emph{GM-MDS conjecture}, due to Dau~\emph{et~al.}, states that if $M$ satisfies the so-called MDS condition, then for any field $\mathbb{F}$ of size $q\geq n+m-1$, there exists an $[n,m]_q$ MDS code whose generator matrix $G$, with entries in $\mathbb{F}$, fits $M$ (i.e., $M$ is the support matrix of $G$). Despite all the attempts by the coding theory community, this conjecture remains still open in general. It was shown, independently by Yan \emph{et al.} and Dau \emph{et al.}, that the GM-MDS conjecture holds if the following conjecture, referred to as the \emph{TM-MDS conjecture}, holds: if $M$ satisfies the MDS condition, then the determinant of a transformation matrix $T$, such that $TV$ fits $M$, is not identically zero, where $V$ is a Vandermonde matrix with distinct parameters. In this work, we generalize the TM-MDS conjecture, and present an algebraic-combinatorial approach based on polynomial-degree reduction for proving this conjecture. Our proof technique's strength is based primarily on reducing inherent combinatorics in the proof. We demonstrate the strength of our technique by proving the TM-MDS conjecture for the cases where the number of rows ($m$) of $M$ is upper bounded by $5$. For this class of special cases of $M$ where the only additional constraint is on $m$, only cases with $m\leq 4$ were previously proven theoretically, and the previously used proof techniques are not applicable to cases with $m > 4$. 
\end{abstract}


\section{Introduction}
In recent years, there has been a growing interest in designing maximum distance separable (MDS) codes with constraints on the support of the codes' generator matrix \cite{YS:2013,DSDY:2013,HHD:2014,DSY:2014,YSZ:2014,HHYD:2014,HTH:2015,DSY:2015,HLH2:2016,HLH:2016,Y:2016}. Such constraints arise in wireless network coding and distributed storage scenarios where each user/server has access to only a subset of the information symbols. Two examples of such scenarios are cooperative data exchange in the presence of an eavesdropper \cite{YS:2013,YSZ:2014}, and simple multiple access networks with link/relay errors \cite{HHYD:2014,DSY:2015}.



Given an $m\times n$ binary (support) matrix $M=(M_{i,j})$ and a field $\mathbb{F}$ of size $q$, the problem is to design an $[n,m]_q$ MDS code with a generator matrix $G=(G_{i,j})$, $G_{i,j}\in\mathbb{F}$ (i.e., all $m\times m$ sub-matrices of $G$ are full-rank) fitting $M$ (i.e., if $M_{i,j}=0$, then $G_{i,j}=0$). Note that for some $M$, there exists no MDS code whose generator matrix fits $M$ (i.e., $M$ is not completable to MDS). Nevertheless, there is a simple condition, known as the MDS condition, which characterizes all matrices that are completable to MDS for sufficiently large fields \cite{DSY:2014,YSZ:2014}. A matrix $M$ satisfies the \emph{MDS condition} if:  
\[|\cup_{i\in I}\hspace{0.25em} \mathrm{supp}(M_i)|\geq n-m+|I|, \forall I\subseteq \{1,\dots,m\}, I\neq \emptyset,\] where $M_i$ is the $i$th row of $M$, and $\mathrm{supp}(M_i)$ is the support of $M_i$. The existence of MDS codes (over sufficiently large fields) whose generator matrix's support satisfies the MDS condition was shown, e.g., in \cite{DSY:2015}, via Edmonds matrix and Hall's marriage theorem. The following conjecture, due to Dau~\emph{et~al.} \cite{DSY:2014}, aims for generalizing this result for small fields.


\begin{conjecture}[GM-MDS Conjecture]
If the matrix $M$ satisfies the MDS condition, then for any field $\mathbb{F}$ of size $q\geq n+m-1$, there exists an $[n,m]_q$ MDS code whose generator matrix $G$ with entries in $\mathbb{F}$ fits the matrix $M$.	
\end{conjecture} 


Notwithstanding all the efforts by the coding theory community, the GM-MDS conjecture remains still open in general. The GM-MDS conjecture and a simplified version of this conjecture where the supports of rows of $M$ have the same size were shown in~\cite{DSY:2015} to be equivalent (using a generalized version of Hall's theorem). Despite this simplification, there are only three classes of special cases for which this conjecture is theoretically proven: (i) the rows of $M$ are divided into three groups, and the rows in each group have the same support \cite{HHYD:2014}; (ii) the size of intersection of the supports of every two rows of $M$ is upper bounded by $1$ \cite{DSY:2015}; and (iii) the number of rows of $M$ is upper bounded by $4$ \cite{Y:2016}. More importantly, the previously used proof techniques are not applicable to more general cases due to the combinatorial explosion.  
 


One possible approach to find a completion $G$ of $M$ to MDS is to leverage the structure of Generalized Reed-Solomon (GRS) codes \cite{YSZ:2014, DSY:2014} which are known to be MDS. Let $N$ be the set of $n$ independent indeterminates $\alpha_1,\dots,\alpha_n$. Let $M$ be an $m\times n$ binary matrix whose rows' supports have the same size, and let $V = V(N)$ be a generic $m\times n$ Vandermonde matrix with parameters $\alpha_1,\dots,\alpha_n$. Let $T = T(M,N)$ be a generic $m\times m$ transformation matrix such that $TV$ fits $M$, and let $G=TV$. If the evaluations $\alpha^{*}_1,\dots,\alpha^{*}_n$ of $\alpha_1,\dots,\alpha_n$ are distinct, then every $m\times m$ sub-matrix of $G$ is non-singular (i.e., $G$ is a generator matrix of a GRS code with evaluation points $\alpha^{*}_1,\dots,\alpha^{*}_n$) so long as $T$ is non-singular. That is, if $T$ is \emph{not generically singular} (i.e., the determinant of $T$ as a multivariate polynomial in variables $\alpha_1,\dots,\alpha_n$ is not identically zero), then for any field $\mathbb{F}$ of size $q\geq n+m-1$, there exists $\alpha^{*}_1,\dots,\alpha^{*}_n\in\mathbb{F}$ such that $TV$ is a generator matrix of an $[n,m]_q$ MDS code, and $TV$ fits $M$. Thus, the GM-MDS conjecture holds if the following conjecture, proposed independently by Dau \emph{et al.}~\cite{DSY:2014} and Yan \emph{et al.}~\cite{YSZ:2014}, holds: 
\begin{conjecture}[TM-MDS Conjecture]\label{conj:TM-MDS}
If $M$ satisfies the MDS condition, then $T(M,N)$ is not generically singular. 
\end{conjecture}  





The contributions of this work are as follows. 
First, we present a generalization of the TM-MDS conjecture for the cases where the supports of rows of $M$ have arbitrary sizes. Then, we present an algebraic-combinatorial approach based on polynomial-degree reduction for proving this conjecture. Our technique's strength is primarily due to reducing the inherent combinatorics in the proof. Specifically, we demonstrate this strength by proving the TM-MDS conjecture for the cases where the number of rows of $M$ is upper bounded by $5$.



\section{Basic Notations and Definitions} 
Let $\mathbb{F}$ be a field. For $n\in\mathbb{N}$, let $\alpha_1,\dots,\alpha_n$ be $n$ independent indeterminates. Let $\mathbb{F}[\alpha_1,\dots,\alpha_n]$ be a ring of multivariate polynomials in variables $\alpha_1,\dots,\alpha_n$ with coefficients in $\mathbb{F}$, and let $(\mathbb{F}[\alpha_1,\dots,\alpha_n])[\alpha]$ be a module of univariate polynomials in variable $\alpha$ with coefficients in $\mathbb{F}[\alpha_1,\dots,\alpha_n]$. 
Fix $m\in\mathbb{N}$ and $n\in\mathbb{N}$ such that $1< m\leq n\leq m (m-1)$. For $k\in \mathbb{N}$, denote $\{1,\dots,k\}$ by $[k]$. Define $m$ polynomials $P_1,\dots,P_m$ of degrees $0\leq d_1,\dots,d_m\leq m-1$ in $(\mathbb{F}[\alpha_1,\dots,\alpha_n])[\alpha]$: 
\begin{equation}\label{eq:PDefinition}
P_i(\alpha)\triangleq \prod_{\gamma\in N_i} (\alpha-\gamma), \hspace{5pt}\forall i\in [m],\vspace{-0.5em}
\end{equation} where $N_i$, the set of roots of $P_i$, is a (proper) subset of $N\triangleq \{\alpha_1,\dots,\alpha_n\}$ of size $d_i$. (Note that the roots of $P_i$ are indeterminates.) For $N_i=\emptyset$ and $d_i=0$, $P_i(\alpha)\triangleq 1$. 

Note that $P_i(\alpha) = \sum_{j\in [m]} C_{i,j} \alpha^{j-1}$, where $\{C_{i,j}\}_{j\in[m]}$ are polynomials in $\mathbb{F}[\alpha_1,\dots,\alpha_n]$. Define $W(P_1,\dots,P_m)\triangleq \det((C_{j,i})_{i,j\in [m]})$. Note that $W(P_1,\dots,P_m)$ is a polynomial in $\mathbb{F}[\alpha_1,\dots,\alpha_n]$. 

%






\begin{definition}
A polynomial $W$ in $\mathbb{F}[\alpha_1,\dots,\alpha_n]$ is \emph{identically zero}, denoted by $W\equiv 0$, if the coefficients of all monomials in the polynomial expansion of $W$ are zero.	
\end{definition}

\begin{definition}
A set $\{P_1,\dots,P_m\}$ of $m$ polynomials of degree $m-1$ has \emph{rectangular property} (RP) if, for some $1< k\leq m$, there exist at least $k$ polynomials in $\{P_1,\dots,P_m\}$ with at least $m-k+1$ common roots. Otherwise, $\{P_1,\dots,P_m\}$ has \emph{non-rectangular property} (NRP). 
\end{definition}

\begin{definition}
A set $\{P_1,\dots,P_m\}$ of $m$ polynomials of degrees $0\leq d_1,\dots,d_m\leq m-1$ has \emph{generalized RP} (GRP) if, for some $1< k\leq m$ and $0\leq l\leq m-k$, there exist at least $k$ polynomials of degrees at most $m-l-1$ in $\{P_1,\dots,P_m\}$ with at least $m-k-l+1$ common roots. Otherwise, $\{P_1,\dots,P_m\}$ has \emph{generalized NRP} (GNRP).
\end{definition}

\section{Main Conjectures and Theorems}


The following conjecture is equivalent to the TM-MDS conjecture (Conjecture~\ref{conj:TM-MDS}). 

\begin{conjecture}\label{conj:Special}
Let $P_1,\dots,P_m$ be $m$ polynomials of degree $m-1$ in $(\mathbb{F}[\alpha_1,\dots,\alpha_n])[\alpha]$. If $W(P_1,\dots,P_m)\equiv 0$, then $\{P_1,\dots,P_m\}$ has RP.
\end{conjecture}


The sketch of proof of the equivalency between the TM-MDS conjecture and Conjecture~\ref{conj:Special} follows. Let $M=(M_{i,j})$ be an $m\times n$ binary matrix. Let $M_i$ be the $i$th row of $M$ and let $\mathrm{supp}(M_i)$ be the support of $M_i$. Let $N_i=\{\alpha_j: j\in [n]\setminus \mathrm{supp}(M_i)\}$, where $\alpha_1,\dots,\alpha_n$ are $n$ independent indeterminates. Suppose that all $N_i$ have the same size. Defining $P_i$ ($=P_i(\alpha)$) as in~\eqref{eq:PDefinition}, it follows that the matrix $M$ satisfies the MDS condition iff $\{P_1,\dots,P_m\}$ has NRP (refer to this as Fact 1). Let $G=(G_{i,j})$ be a generic $m\times n$ generator matrix of a Generalized Reed-Solomon (GRS) code with evaluation points $\alpha_1,\dots,\alpha_n$ such that $G$ fits $M$ (i.e., if $M_{i,j}=0$, then $G_{i,j}=0$). Let $V=(V_{i,j})$ be a generic $m\times n$ Vandermonde matrix with parameters $\alpha_1,\dots,\alpha_n$, and let $T=(T_{i,j})$ be a generic $m\times m$ transformation matrix such that $G=TV$. Taking $V_{i,j}=\alpha^{j-1}_i$ and $T_{i,j}=C_{j,i}$, where $P_i(\alpha) = \sum_{j\in [m]} C_{i,j} \alpha^{j-1}$, it follows that $G=TV$ (for more details, see \cite{YSZ:2014}). Since $\det(T)=W(P_1,\dots,P_{m})$, then $\det(T)\not\equiv 0$ (i.e., $T$ is not generically singular) iff $W(P_1,\dots,P_m)\not\equiv 0$ (refer to this as Fact 2). By Facts 1 and 2, the TM-MDS conjecture and Conjecture~\ref{conj:Special} are equivalent. 

%





In the following, we propose a new conjecture which generalizes Conjecture~\ref{conj:Special} for the cases where the degrees $d_1,\dots,d_m$ of polynomials $P_1,\dots,P_m$ are arbitrary. (Conjecture~\ref{conj:General} is equivalent to a generalized version of the TM-MDS conjecture where the supports of rows of $M$ have arbitrary sizes.) 

\begin{conjecture}\label{conj:General}
Let $P_1,\dots,P_m$ be $m$ polynomials of arbitrary degrees $0\leq d_1,\dots,d_m\leq m-1$ in $(\mathbb{F}[\alpha_1,\dots,\alpha_n])[\alpha]$. If $W(P_1,\dots,P_m)\equiv 0$, then $\{P_1,\dots,P_m\}$ has GRP.\end{conjecture}

If $d_i<m-1$ for all $i\in [m]$, then Conjecture~\ref{conj:General} holds trivially: (i) $W(P_1,\dots,P_m)\equiv 0$ since $C_{i,m}= 0$ for all $i\in [m]$, and (ii) $\{P_1,\dots,P_m\}$ has GRP since for $k=m$ and $l=1$, there exist $k$ polynomials of degrees at most $m-l-1$ in $\{P_1,\dots,P_m\}$ with at least $m-k-l+1$ common roots. Hereafter, w.l.o.g., we assume $d_i= m-1$ for some $i\in [m]$. 

The following theorems, which are our main results, prove the GM-MDS conjecture for $m\leq 5$. More specifically, Theorems~\ref{thm:FirstSp},~\ref{thm:SecondSp}, and~\ref{thm:ThirdSp} settle Conjecture~\ref{conj:General} (and so Conjecture~\ref{conj:Special}) for $m\leq 4$, and Theorem~\ref{thm:FourthSp} settles Conjecture~\ref{conj:Special} for $m=5$. 

\begin{theorem}\label{thm:FirstSp}
For any $P_1,P_2$ such that $0\leq d_1\leq d_2=1$, if $W(P_1,P_2)\equiv 0$, then $\{P_1,P_2\}$ has GRP. 	
\end{theorem}

\begin{theorem}\label{thm:SecondSp}
For any $P_1,P_2,P_3$ such that $0\leq d_1\leq d_2\leq d_3=2$, if $W(P_1,P_2,P_3)\equiv 0$, then $\{P_1,P_2,P_3\}$ has GRP.
\end{theorem}  

\begin{theorem}\label{thm:ThirdSp}
For any $P_1,\dots,P_4$ such that $0\leq d_1\leq \dots\leq d_4=3$, if $W(P_1,\dots,P_4)\equiv 0$, then $\{P_1,\dots,P_4\}$ has GRP.
\end{theorem}  


\begin{theorem}\label{thm:FourthSp}
For any $P_1,\dots,P_5$ such that $d_1=\dots=d_5=4$, if $W(P_1,\dots,P_5)\equiv 0$, then $\{P_1,\dots,P_5\}$ has GRP. 	
\end{theorem}

\section{Main Ideas and Lemmas}
In this section, we explain the main ideas and state the useful lemmas for the proofs of our main results. 


Consider an arbitrary set $\{P_i\}$ ($=\{P_i\}_{1\leq i\leq m}$) of $m$ polynomials $P_i$ (with the sets of roots $N_i$) such that $0\leq d_i\leq m-1$ for all $i\in [m]$, and $d_i=m-1$ for some $i\in [m]$. Define a class of \emph{reduction processes} over $\{P_i\}$, where any process in this class is associated with a unique \emph{reduction set} $R\subseteq N$, and it reduces $P_i(\alpha)=\prod_{
\gamma\in N_i} (\alpha-\gamma)$ to $\tilde{P}_i(\alpha)\triangleq \prod_{\gamma\in N_i\setminus R} (\alpha-\gamma)$. Let $\tilde{d}_i\triangleq \mathrm{deg}(\tilde{P}_i)$. Note that $\tilde{d}_i=d_i-|N_i\cap R|$. Restrict your attention to those reduction sets $R$ such that $\tilde{d}_j=m-1$ for some $j\in [m]$, and $\tilde{d}_i<m-1$ for all $i\in[m]\setminus \{j\}$. Such $R$ are referred to as \emph{acceptable}. For any acceptable reduction set, w.l.o.g., assume that $\tilde{d}_i<m-1$ for all $i\in [m-1]$ and $\tilde{d}_m=m-1$ (and so, $d_m=m-1$ since $\tilde{d}_m\leq d_m\leq m-1$). For any (acceptable) reduction set, the following result holds. 

\begin{lemma}\label{lem:RW}
If $W(P_1,\dots,P_m)\equiv 0$, then $W(\tilde{P}_1,\dots,\tilde{P}_{m-1})\equiv 0$. 
\end{lemma}

\begin{proof}
Consider the resulting $\{\tilde{P}_i\}$ from $\{P_i\}$ for an arbitrary (acceptable) reduction set $R=\{\alpha_{r_1},\dots,\alpha_{r_{|R|}}\}$. Let $r_R \triangleq \{r_1,\dots,r_{|R|}\}$. For any $r\in r_R$, let $n_{r}$ be the number of polynomials $P_i$ in $\{P_i\}$ such that $\alpha_{r}\in N_i$, and let $n_R\triangleq \{n_{r_1},\dots,n_{r_{|R|}}\}$. Let $W^{(n_R)}(P_1,\dots,P_m)$ be the resulting polynomial from $W(P_1,\dots,P_m)$ by taking derivative $n_r$ times with respect to each variable $\alpha_r\in R$. (Since $W(P_1,\dots,P_m)=\det((C_{j,i})_{i,j\in [m]})$, and $C_{j,i}$ is the sum of monomials $(-1)^{e_1+\dots+e_n}\alpha^{e_1}_1\cdots \alpha^{e_n}_n$ for some $\{e_1,\dots,e_n\}\in \{0,1\}^{n}$ (depending on $i,j$), then the derivatives of $C_{j,i}$ with respect to any variable $\alpha_r$ are independent of $\mathbb{F}$.) Note that $W^{(n_R)}(P_1,\dots,P_m)=(-1)^{n_{r_1}+\dots+n_{r_{|R|}}} W(\tilde{P}_1,\dots,\tilde{P}_m)$ (by using the Leibniz formula for determinant), and $W(\tilde{P}_1,\dots,\tilde{P}_m)= W(\tilde{P}_1,\dots,\tilde{P}_{m-1})$ (since $\tilde{d}_i<m-1$ and $\tilde{C}_{i,m}= 0$ for all $i\in [m-1]$, and $\tilde{d}_m=m-1$ and $\tilde{C}_{m,m} = 1$, where $\tilde{P}_i(\alpha) = \sum_{j\in [m]} \tilde{C}_{i,j} \alpha^{j-1}$). Since $W(P_1,\dots,P_m)\equiv 0$ (by assumption), then $W^{(n_R)}({P}_1,\dots,{P}_m)\equiv 0$. Thus, $W(\tilde{P}_1,\dots,\tilde{P}_{m-1})\equiv 0$.
\end{proof} 

Lemma~\ref{lem:RW} enables us to use an inductive argument towards the proof of Conjecture~\ref{conj:General} as follows. Suppose that Conjecture~\ref{conj:General} holds for any $1< l\leq m-1$, i.e., for any $\{P_1,\dots,P_l\}$ such that $0\leq d_i\leq l-1$ for all $i\in [l]$ and $d_i=l-1$ for some $i\in [l]$, if $W(P_1,\dots,P_l)\equiv 0$, then $\{P_1,\dots,P_l\}$ has GRP. We need to prove that for any $\{P_1,\dots,P_m\}$ such that $0\leq d_i\leq m-1$ for all $i\in [m]$ and $d_i=m-1$ for some $i\in [m]$, if $W(P_1,\dots,P_m)\equiv 0$, then $\{P_1,\dots,P_m\}$ has GRP. The proof follows by contradiction. Assume that $W(P_1,\dots,P_m)\equiv 0$ and $\{P_1,\dots,P_m\}$ does not have GRP. Consider the resulting $\{\tilde{P}_1,\dots,\tilde{P}_m\}$ from $\{P_1,\dots,P_m\}$ for an (acceptable) reduction set such that $\{\tilde{P}_1,\dots,\tilde{P}_{m-1}\}$ has GNRP. By definition, $\tilde{d}_i<m-1$ for all $i\in [m-1]$. Since $W(P_1,\dots,P_m)\equiv 0$ (by assumption), then $W(\tilde{P}_1,\dots,\tilde{P}_{m-1})\equiv 0$ (by Lemma~\ref{lem:RW}), and so, $\{\tilde{P}_1,\dots,\tilde{P}_{m-1}\}$ has GRP (by the induction hypothesis), yielding a contradiction. Our goal is thus to devise an (acceptable) reduction process such that if $\{P_1,\dots,P_m\}$ has GNRP, then so does $\{\tilde{P}_1,\dots,\tilde{P}_{m-1}\}$. The problem of designing such a process is still open in general. In the following, we propose a simple yet powerful reduction process which solves this problem for $m\leq 4$ and $0\leq d_i\leq m-1$ for all $i\in [m]$, and for $m=5$ and $d_i=m-1$ for all $i\in [m]$. 






From now on, we assume that $\{P_i\}$ ($=\{P_i\}_{1\leq i\leq m}$) is a set of $m$ polynomials $P_i$ (with the sets of roots $N_i$) such that $0\leq d_i\leq m-1 $ for all $i\in [m-1]$, and $d_m=m-1$. 

\begin{definition}
A subset $S\subset N$ is an \emph{$(r,s)$-subset} in a subset $\mathcal{Q}$ of $\{P_i\}$ if $S$ belongs to $r$ polynomials in $\mathcal{Q}$ (i.e., there exist $r$ polynomials $P_i$ in $\mathcal{Q}$ such that $S\subset N_i$), and $|S|=s$. Moreover, an $(r,s)$-subset has \emph{higher order} than an $(r^{\star},s^{\star})$-subset if $r+s>r^{\star}+s^{\star}$, or $r+s=r^{\star}+s^{\star}$ and $r>r^{\star}$.
\end{definition}

The following lemma gives the intuition behind the definition of $(r,s)$-subsets. 

\begin{lemma}\label{lem:GRPOrder}
If $\{P_i\}$ has GNRP, then there exists no $(r,s)$-subset in $\{P_i\}$ such that $r+s> m$. 
\end{lemma}

\begin{proof}
The proof is straightforward and follows from the definitions (and hence omitted).
\end{proof}



Intuitively, for any (acceptable) reduction set $R$, any highest-order $(r,s)$-subset $S$, if not broken (i.e., $S\cap R = \emptyset$), is the most likely to cause rectangularity in $\{\tilde{P}_i\}$ for any $\{P_i\}$ with non-rectangular property. This is the main idea of the proposed reduction process. 

\begin{definition}
An element $\beta$ of a subset $S\subset N$ is \emph{removable} if $\beta$ is a root of some but not all polynomials of degree $m-1$. 
\end{definition}

\begin{definition}
A subset $S\subset N$ is \emph{weakly reducible} if $S$ belongs to a polynomial of degree $m-1$, and $S$ has a removable element. 
\end{definition}

\begin{definition}
A weakly reducible $(r,s)$-subset $S$ is \emph{strongly reducible} if no other weakly reducible $(r^{\star},s^{\star})$-subset has higher order than $S$.
\end{definition}





\subsubsection*{Proposed Reduction Process}
Given $\{P_i\}$, choose an arbitrary strongly reducible subset $S$ in $\{P_i\}$, and choose an arbitrary removable element of $S$, say $\beta$, such that no other removable element of $S$, when compared to $\beta$, belongs to more polynomials of degree $m-1$ in $\{P_i\}$. Break $S$ via removing $\beta$ from the sets $N_i$ of roots of all polynomials $P_i$, and update all polynomials $P_i$ via replacing $N_i$ by $N_i\setminus\{\beta\}$. Repeat this process (in rounds) over the resulting $\{P_i\}$ if there exist more than one polynomial of degree $m-1$. Otherwise, terminate the process, and return the resulting $\{P_i\}$ denoted by $\{\tilde{P}_i\}$. 

Note that if $\{P_i\}$ has GNRP initially, then (i) in each round of the process, such $\beta$ exists, and (ii) the process terminates eventually. Otherwise, there must exist two (or more) identical polynomials of degree $m-1$ in $\{P_i\}$ (and hence $\{P_i\}$ has GRP), which is a contradiction. 




Consider an arbitrary run of the reduction process over $\{P_i\}$ and its corresponding $\{\tilde{P}_i\}$. Let $R$ be the set of the roots that the reduction process removes over the rounds. (Note that, due to the arbitrary choices in the reduction process, $R$ may or may not be unique.) Hereafter, for any such $R$, assume, w.l.o.g., that the (initial) indexing of polynomials in $\{P_i\}$ is such that $\tilde{d}_1\leq \dots\leq \tilde{d}_{m-1}<\tilde{d}_m$ ($=d_m$) and $\tilde{P}_m=P_m$, and denote $\{P_i\}_{1\leq i\leq m-1}$ by $\{P_i\}_{\star}$. 



The proofs of our main theorems rely on the following properties of the proposed reduction process. 

%



\begin{lemma}\label{lem:MaxHD}
If $\{P_i\}$ has GNRP, then any $(r,s)$-subset in $\{P_i\}_{\star}$ such that $r+s=m$ belongs to a polynomial $P_i$ in $\{P_i\}_{\star}$ such that $d_i=m-1$.
\end{lemma}

\begin{proof}
Let $S$ be an arbitrary $(r,s)$-subset in $\{P_i\}_{\star}$ such that $r+s=m$. Let $\mathcal{Q}$ be the set of all polynomials $P_i$ in $\{P_i\}_{\star}$ such that $d_i< m-1$. Note that $\mathcal{Q}$ is a set of at most $m-1$ polynomials of degrees at most $m-2$. Since $\{P_i\}$ has GNRP (by assumption), then $\{P_i\}_{\star}$ (and hence $\mathcal{Q}$) has GNRP. Thus, there exists no $(r^{\star},s^{\star})$-subset in $\mathcal{Q}$ such that $r^{\star}+s^{\star}>m-1$ (by Lemma~\ref{lem:GRPOrder}). Suppose that $S$ belongs to no polynomial $P_i$ in $\{P_i\}_{\star}\setminus \mathcal{Q}$. Then, $S$ is an $(r,s)$-subset in $\mathcal{Q}$ such that $r+s=m$. This is, however, a contradiction. Thus, $S$ belongs to a polynomial $P_i$ in $\{P_i\}_{\star}\setminus \mathcal{Q}$. 
\end{proof}

\begin{lemma}\label{lem:ProperCol}
If $\{P_i\}$ has GNRP, then any $(r,s)$-subset in $\{P_i\}_{\star}$ such that $r+s=m$ is weakly reducible. 
\end{lemma}

\begin{proof}
Let $S$ be an arbitrary $(r,s)$-subset in $\{P_i\}_{\star}$ such that $r+s=m$. Note that $S$ belongs to a polynomial in $\{P_i\}_{\star}$ (and so $\{P_i\}$) of degree $m-1$ (by Lemma~\ref{lem:MaxHD}). Note, also, that $P_m$ has degree $m-1$. Thus, if there exists $\beta\in S$ such that $\beta\not\in N_m$, then $S$ is weakly reducible since $\beta$ is removable (by definition). Otherwise, if $S\subseteq N_m$, then $S$ is an $(r+1,s)$-subset in $\{P_i\}$. Since $r+1+s=m+1>m$, then $\{P_i\}$ has GRP (by Lemma~\ref{lem:GRPOrder}), yielding a contradiction.\end{proof}

\begin{lemma}\label{lem:DisjointMax}
If $\{P_i\}$ has GNRP, then the strongly reducible $(r,s)$-subsets in $\{P_i\}_{\star}$ such that $r+s=m$ belong to disjoint subsets of $\{P_i\}_{\star}$. 	
\end{lemma}

\begin{proof}
Let $S_1$ and $S_2$ be two arbitrary strongly reducible $(r,s)$-subsets in $\{P_i\}_{\star}$ such that $r+s=m$. Let $\mathcal{Q}_1$ or $\mathcal{Q}_2$ be the set of $r$ polynomials $P_i$ in $\{P_i\}_{\star}$ such that $S_1$ or $S_2$ belongs to $P_i$, respectively. Note that $0\leq |\mathcal{Q}_1\cap\mathcal{Q}_2|\leq r$. First, suppose that $|\mathcal{Q}_1\cap\mathcal{Q}_2|=r$. Then, $S=S_1\cup S_2$ is an $(r,|S_1\cup S_2|)$-subset in $\{P_i\}_{\star}$. Since $r+|S_1\cup S_2|>r+s=m$, then $\{P_i\}$ has GRP (by Lemma~\ref{lem:GRPOrder}), and hence a contradiction. Next, suppose that $0<|\mathcal{Q}_1\cap \mathcal{Q}_2|<r$. We consider two cases. First, suppose that $|S_1\cap S_2|<m-2r+|\mathcal{Q}_1\cap\mathcal{Q}_2|$. Then, $S=S_1\cup S_2$ is a $(|\mathcal{Q}_1\cap\mathcal{Q}_2|,2s-|S_1\cap S_2|)$-subset in $\{P_i\}_{\star}$. Let $r^{\star}=|\mathcal{Q}_1\cap\mathcal{Q}_2|$ and $s^{\star}=2s-|S_1\cap S_2|$. Since $r^{\star}+s^{\star}=|\mathcal{Q}_1\cap\mathcal{Q}_2|+m-2r-|S_1\cap S_2|>m$, then $\{P_i\}$ has GRP (by Lemma~\ref{lem:GRPOrder}), which is a contradiction. Next, suppose that $|S_1\cap S_2|\geq m-2r+|\mathcal{Q}_1\cap \mathcal{Q}_2|$. Then, $S=S_1\cap S_2$ is a $(2r-|\mathcal{Q}_1\cap\mathcal{Q}_2|,|S_1\cap S_2|)$-subset in $\{P_i\}_{\star}$. Let $r^{\star}=2r-|\mathcal{Q}_1\cap\mathcal{Q}_2|$ and $s^{\star}=|S_1\cap S_2|$. Note that $r^{\star}+s^{\star}=2r-|\mathcal{Q}_1\cap\mathcal{Q}_2|+|S_1\cap S_2|\geq m$. If $r^{\star}+s^{\star}>m$, then $\{P_i\}$ has GRP (by Lemma~\ref{lem:GRPOrder}), and hence a contradiction. If $r^{\star}+s^{\star}=m$, then $S$ is weakly reducible (by Lemma~\ref{lem:ProperCol}), and $S$ has higher order than $S_1$ and $S_2$ since $r^{\star}+s^{\star}=r+s$ and $r^{\star}>r$. This is also a contradiction since $S_1$ and $S_2$ are strongly reducible (by assumption). Thus, $|\mathcal{Q}_1\cap \mathcal{Q}_2|=0$. 
\end{proof}

\begin{lemma}\label{lem:MaxCol}
For $m=2,3,4$ and $0\leq d_i\leq m-1$ for all $i\in [m]$, and for $m=5$ and $d_i=m-1$ for all $i\in [m]$, if $\{P_i\}$ has GNRP, then the reduction process breaks any strongly reducible $(r,s)$-subset in $\{P_i\}_{\star}$ such that $r+s=m$.
\end{lemma}

\begin{proof}
Let $S$ be an arbitrary strongly reducible $(r,s)$-subset in $\{P_i\}_{\star}$ such that $r+s=m$. Since $\{P_i\}$ has GNRP, $S$ belongs to a polynomial $P_i$ in $\{P_i\}_{\star}$ of degree $m-1$ (by Lemma~\ref{lem:MaxHD}) and no other strongly reducible $(r,s)$-subset in $\{P_i\}_{\star}$ belongs to $P_i$ (by Lemma~\ref{lem:DisjointMax}). Moreover, for any $m=2,3,4$ and any $0\leq d_1\leq \dots\leq d_m=m-1$, there exists no other $(r,s)$-subset in $\{P_i\}_{\star}$ such that $r+s=m$ (otherwise, $\{P_i\}$ has GRP). Thus $S$ must be broken to reduce $P_i$. 

For $m=5$ and $d_1=\dots=d_5=4$, $S$ is either a $(4,1)$- or $(3,2)$- or $(2,3)$-subset in $\{P_i\}_{\star}$. First, suppose that $S$ is a $(4,1)$- or $(3,2)$-subset in $\{P_i\}_{\star}$. Since there exists no other $(4,1)$- or $(3,2)$-subset in $\{P_i\}_{\star}$ (otherwise, $\{P_i\}$ has GRP), then $S$ must be broken to reduce $P_i$. Next, suppose that $S$ is a $(2,3)$-subset in $\{P_i\}_{\star}$. Let $\mathcal{Q}$ be the set of two polynomials in $\{P_i\}_{\star}$, say $P_1$ and $P_2$, such that $S$ belongs to both $P_1$ and $P_2$. Let $T$ be an arbitrary (if any) strongly reducible $(2,3)$-subset in $\{P_i\}_{\star}\setminus \mathcal{Q}$. If $T$ does not exist, then $S$ must be broken to reduce both $P_1$ and $P_2$. If $T$ exists, no element of $T$ is a common root of both $P_1$ and $P_2$ (otherwise, there exists a strongly reducible $(4,1)$-subset in $\{P_i\}_{\star}$, which is a contradiction since $S$ is a strongly reducible $(2,3)$-subset). Since $\{P_i\}$ has GNRP, then there exists no other strongly reducible $(2,3)$-subset in $\{P_i\}_{\star}$ (by Lemma~\ref{lem:DisjointMax}), and breaking $T$ cannot reduce both $P_1$ and $P_2$ simultaneously. Thus, $S$ must be broken to reduce $P_1$ or $P_2$ (or both).
\end{proof}

\section{Proofs of Main Theorems}
In this section, we prove our main theorems. For simplicity, we denote the degree-set of polynomials $P_1,\dots,P_m$ and $\tilde{P}_1,\dots,\tilde{P}_m$ by $(d_1,\dots,d_m)$ and $(\tilde{d}_1,\dots,\tilde{d}_m)$, respectively. 
%
%
%

\begin{proof}[Proof of Theorem~\ref{thm:FirstSp}]
Assume that $W(P_1,P_2)\equiv 0$. If $(d_1,d_2)=(0,1)$, then $W(P_1,P_2)=1\not\equiv 0$, which is a contradiction. If $(d_1,d_2)=(1,1)$, then $W(P_1,P_2)=P_1-P_2$. Thus, $P_1=P_2$, i.e., $\{P_1,P_2\}$ has RP (and hence GRP).   	
\end{proof}

\begin{proof}[Proof of Theorem~\ref{thm:SecondSp}]
The proof follows by contradiction. Assume that $W(P_1,P_2,P_3)\equiv 0$, and $\{P_1,P_2,P_3\}$ has GNRP. If $(d_1,d_2,d_3)=(2,2,2)$, then $(\tilde{d}_1,\tilde{d}_2,\tilde{d}_3)=(1,1,2)$ (since the reduction process either reduces both $P_1$ and $P_2$ simultaneously, or it first reduces one, and then reduces the other one). Since $W(P_1,P_2,P_3)\equiv 0$ (by assumption), then $W(\tilde{P}_1,\tilde{P}_2)\equiv 0$ (by Lemma~\ref{lem:RW}). Thus, $\{\tilde{P}_1,\tilde{P}_2\}$ has GRP (by Theorem~\ref{thm:FirstSp}), i.e., there exists a $(2,1)$-subset $S$ in $\{\tilde{P}_1,\tilde{P}_2\}$. Thus, $S$ is a strongly reducible $(r,s)$-subset in $\{P_1,P_2\}$ such that $r+s=m=3$ (by Lemmas~\ref{lem:ProperCol} and~\ref{lem:GRPOrder}), and it must have been broken (by Lemma~\ref{lem:MaxCol}), yielding a contradiction.

If $(d_1,d_2,d_3)=(1,2,2)$, then $(\tilde{d}_1,\tilde{d}_2,\tilde{d}_3)\in \{(1,1,2),(0,1,2)\}$ (since the reduction process must reduce $P_2$, and reducing $P_2$ may or may not reduce $P_1$). If $(\tilde{d}_1,\tilde{d}_2,\tilde{d}_3)=(1,1,2)$, then $\{\tilde{P}_1,\tilde{P_2}\}$ has GRP, yielding a contradiction as before. If $(\tilde{d}_1,\tilde{d}_2,\tilde{d}_3)=(0,1,2)$, then $W(\tilde{P}_1,\tilde{P}_2)=1\not\equiv 0$. Thus, $W(P_1,P_2,P_3)\not\equiv 0$ (by Lemma~\ref{lem:RW}), which is again a contradiction. 

If $(d_1,d_2,d_3)=(0,2,2)$, then $(\tilde{d}_1,\tilde{d}_2,\tilde{d}_3)=(0,1,2)$ (since the reduction process must reduce $P_2$, and reducing $P_2$ does not reduce $P_1$). Since $W(\tilde{P}_1,\tilde{P}_2)=1\not\equiv 0$, then $W(P_1,P_2,P_3)\not\equiv 0$ (by Lemma~\ref{lem:RW}), yielding a contradiction. If $(d_1,d_2,d_3)=(1,1,2)$, then $(\tilde{d}_1,\tilde{d}_2,\tilde{d}_3)=(1,1,2)$ (since the reduction process does not reduce $P_1$ and $P_2$). Thus, $\{\tilde{P}_1,\tilde{P}_2\}$ ($=\{P_1,P_2\}$) has GRP (by the same argument as before), which is a contradiction. If $(d_1,d_2,d_3)=(0,1,2)$, then $W(P_1,P_2,P_3)=1\not\equiv 0$, yielding a contradiction. If $(d_1,d_2,d_3)=(0,0,2)$, then $P_1=P_2=1$. Thus, $\{P_1,P_2\}$ has GRP, again a contradiction.  
\end{proof}


\begin{proof}[Proof of Theorem~\ref{thm:ThirdSp}]
Due to the lack of space, we only give the proofs for the cases of $(d_1,\dots,d_4)\in\{(3,3,3,3),$ $(2,3,3,3),$ $(1,3,3,3),$ $(2,2,3,3)\}$. (The proofs of the rest of the cases follow the exact same lines.) The proof is by way of contradiction. Assume that $W(P_1,\dots,P_4)\equiv 0$, and $\{P_1,\dots,P_4\}$ has GNRP. Since $W(\tilde{P}_1,\tilde{P}_2,\tilde{P}_3)\equiv 0$ (by Lemma~\ref{lem:RW}), then $\{\tilde{P}_1,\tilde{P}_2,\tilde{P}_3\}$ has GRP (by Theorem~\ref{thm:SecondSp}). 

First, consider $(d_1,\dots,d_4)=(3,3,3,3)$. By the procedure of the reduction process, $(\tilde{d}_1,\dots,\tilde{d}_4)\in\{(2,2,2,3),(1,2,2,3)\}$. Since $\{\tilde{P}_1,\tilde{P}_2,\tilde{P}_3\}$ has GRP, either there exists a $(3,1)$-subset $S_1$, or if $S_1$ does not exist, there exists a $(2,2)$-subset $S_2$, in $\{\tilde{P}_1,\tilde{P}_2,\tilde{P}_3\}$. Since $S_1$ (or $S_2$) is a strongly reducible $(r,s)$-subset in $\{P_1,P_2,P_3\}$ such that $r+s=m=4$ (by Lemmas~\ref{lem:ProperCol} and~\ref{lem:GRPOrder}), $S_1$ (or $S_2$) must have been broken (by Lemma~\ref{lem:MaxCol}), which is a contradiction.


Second, consider $(d_1,\dots,d_4)=(2,3,3,3)$. Then, $(\tilde{d}_1,\dots,\tilde{d}_4)\in\{(2,2,2,3),(1,2,2,3),(0,2,2,3)\}$. For any of these cases, by the same arguments as for the previous case, we arrive at a contradiction. 

Next, consider $(d_1,\dots,d_4)=(1,3,3,3)$. Then, $(\tilde{d}_1,\dots,\tilde{d}_4)\in\{(1,2,2,3),(0,2,2,3),(0,1,2,3)\}$. For the cases of $(\tilde{d}_1,\dots,\tilde{d}_4)\in\{(1,2,2,3),(0,2,2,3)\}$, similar to the previous cases, we reach a contradiction. For the case of $(\tilde{d}_1,\dots,\tilde{d}_4)=(0,1,2,3)$, it follows that $W(\tilde{P}_1,\tilde{P}_2,\tilde{P}_3)=1\not\equiv 0$, which is again a contradiction. 

Lastly, consider $(d_1,\dots,d_4)=(2,2,3,3)$. Then, $(\tilde{d}_1,\dots,\tilde{d}_{4})\in\{(2,2,2,3),(1,2,2,3),(1,1,2,3)\}$. For the cases of $(\tilde{d}_1,\dots,\tilde{d}_{4})\in\{(2,2,2,3),(1,2,2,3)\}$, following the exact same lines as above yields a contradiction. Now, consider the case of $(\tilde{d}_1,\dots,\tilde{d}_{4})=(1,1,2,3)$. Since $\tilde{d}_1=d_1-1$, $\tilde{d}_2=d_2-1$, and $\tilde{d}_3=d_3-1$, then reducing $P_3$ must have reduced $P_1$ and $P_2$ simultaneously. Thus, there exists a $(3,1)$-subset $\{\beta_1\}$ in $\{P_1,P_2,P_3\}$. Since $\{\tilde{P}_1,\tilde{P}_2,\tilde{P}_3\}$ has GRP, there also exists a $(2,1)$-subset $\{\beta_2\}$ in $\{\tilde{P}_1,\tilde{P}_2\}$. Thus, $\{P_1,P_2\}$ has GRP since $\{\beta_1,\beta_2\}$ is a $(2,2)$-subset in $\{P_1,P_2\}$, yielding a contradiction.
\end{proof} 



\begin{proof}[Proof of Theorem~\ref{thm:FourthSp}]
The proof follows by contradiction. Assume that $W(P_1,\dots,P_5)\equiv 0$, and $\{P_1,\dots,P_5\}$ has GNRP. Since $W(\tilde{P}_1,\dots,\tilde{P}_4)\equiv 0$ (by Lemma~\ref{lem:RW}), then $\{\tilde{P}_1,\dots,\tilde{P}_4\}$ has GRP (by Theorem~\ref{thm:ThirdSp}). By the procedure of the reduction process, $(\tilde{d}_1,\dots,\tilde{d}_5)$ $\in\{(3,3,3,3,4),$ $(2,3,3,3,4),$ $(1,3,3,3,4),$ $(2,2,3,3,4)\}$. Consider any of the cases of $(\tilde{d}_1,\dots,\tilde{d}_5)\in$ $\{(3,3,3,3,4),$ $(2,3,3,3,4),$ $(1,3,3,3,4)\}$. Since $\{\tilde{P}_1,\dots,\tilde{P}_4\}$ has GRP, there exists either a $(4,1)$-subset $S_1$, or a $(3,2)$-subset $S_2$ (if $S_1$ does not exist), or a $(2,3)$-subset $S_3$ (if neither $S_1$ nor $S_2$ exists), in $\{\tilde{P}_1,\dots,\tilde{P}_4\}$. Since $S_1$ (or $S_2$ or $S_3$) is a strongly reducible $(r,s)$-subset in $\{P_1,\dots,P_4\}$ such that $r+s=m=5$ (by Lemmas~\ref{lem:ProperCol} and~\ref{lem:GRPOrder}), it must have been broken by the reduction process (by Lemma~\ref{lem:MaxCol}), which is a contradiction. 

Now, consider the case of $(\tilde{d}_1,\dots,\tilde{d}_5)=(2,2,3,3,4)$. Since $\{\tilde{P}_1,\dots,\tilde{P}_4\}$ has GRP, there exists either a $(4,1)$-subset $S_1$, or a $(3,2)$-subset $S_2$, or a $(2,3)$-subset $S_3$, in $\{\tilde{P}_1,\dots,\tilde{P}_4\}$, or if neither of $S_1$, $S_2$, and $S_3$ exists, there exists a $(2,2)$-subset $S_4$ in $\{\tilde{P}_1,\tilde{P}_2\}$. If $S_1$ (or $S_2$ or $S_3$) exists, then it is a strongly reducible $(r,s)$-subset in $\{P_1,\dots,P_4\}$ such that $r+s=m=5$ (by Lemmas~\ref{lem:ProperCol} and~\ref{lem:GRPOrder}), and it must have been broken (by Lemma~\ref{lem:MaxCol}), yielding a contradiction. If neither of $S_1$, $S_2$, and $S_3$ exists, but $S_4$ exists, then there exists a $(2,2)$-subset $\{\beta_1,\beta_2\}$ in $\{\tilde{P}_1,\tilde{P}_2\}$. Since $\tilde{d}_1=d_1-2$, $\tilde{d}_2=d_2-2$, $\tilde{d}_3=d_3-1$, and $\tilde{d}_4=d_4-1$, either (i) $P_3$ and $P_4$ are reduced separately, and reducing $P_3$ and reducing $P_4$ both have reduced $P_1$ and $P_2$ simultaneously, or (ii) $P_1$ and $P_2$ are reduced simultaneously (without reducing $P_3$ or $P_4$), and reducing $P_3$ has reduced $P_1$ (or $P_2$) but not $P_4$, and reducing $P_4$ has reduced $P_2$ (or $P_1$) but not $P_3$. 

First, consider the case (i). Since reducing $P_3$ has reduced both $P_1$ and $P_2$, there exists a $(3,1)$-subset $\{\beta_3\}$ in $\{P_1,P_2,P_3\}$ such that $\beta_3\neq \beta_1,\beta_2$. Similarly, there exists a $(3,1)$-subset $\{\beta_4\}$ in $\{P_1,P_2,P_4\}$ such that $\beta_4\neq \beta_1,\beta_2$. Note that $\beta_3\neq \beta_4$ since otherwise, $\{\beta_3\}$ or $\{\beta_4\}$ is a $(4,1)$-subset in $\{P_1,\dots,P_4\}$, which is a contradiction. Thus, there exists a $(2,4)$-subset $\{\beta_1,\beta_2,\beta_3,\beta_4\}$ in $\{P_1,P_2\}$, i.e., $\{P_1,P_2\}$ has GRP, yielding a contradiction again. 



Next, consider the case (ii). Since $P_1$ and $P_2$ are reduced simultaneously, there exists a $(2,1)$-subset $\{\beta_3\}$ in $\{P_1,P_2\}$ such that $\beta_3\neq \beta_1,\beta_2$. Thus, $\{\beta_1,\beta_2,\beta_3\}$ is a $(2,3)$-subset in $\{P_1,P_2\}$. Note, also, that none of the elements $\beta_1$, $\beta_2$, and $\beta_3$ is a root of $P_3$ or $P_4$. This comes from two facts: (a) if $\beta_3$ is a root of $P_3$ or $P_4$, then reducing $P_1$ and $P_2$ via removing $\beta_3$ must have reduced $P_3$ or $P_4$, which is a contradiction; and (b) if there exists $\beta\in\{\beta_1,\beta_2\}$ such that $\beta$ is a root of $P_3$ or $P_4$, then no other element of $\{\beta_1,\beta_2,\beta_3\}$ belongs, when compared to $\beta$, to more polynomials of degree $4$ in $\{P_1,\dots,P_5\}$ (since $\{\beta\}$ is a $(3,1)$-subset and there exists no $(4,1)$-subset). Thus, $P_1$ and $P_2$ must have been reduced via removing $\beta$, which yields reducing $P_3$ or $P_4$, and hence a contradiction.

Since reducing $P_3$ has reduced $P_1$ (or $P_2$) and reducing $P_4$ has reduced $P_2$ (or $P_1$), then there exist a $(2,1)$-subset $\{\beta_4\}$ in $\{P_1,P_3\}$ and a $(2,1)$-subset $\{\beta_5\}$ in $\{P_2,P_4\}$ such that $\beta_4\neq \beta_5$. Note that $\beta_4$  or $\beta_5$ is not a root of $P_4$ or $P_3$, respectively (otherwise, reducing $P_3$ (or $P_4$) via removing $\beta_4$ (or $\beta_5$) must have reduced $P_4$ (or $P_3$), yielding a contradiction). Thus, $N_1=\{\beta_1,\beta_2,\beta_3,\beta_4\}$ and $N_2=\{\beta_1,\beta_2,\beta_3,\beta_5\}$. Since there is no $(3,3)$-subset in $\{P_3,P_4,P_5\}$, then $P_3$ has a root $\beta_6$ ($\neq \beta_1,\beta_2,\beta_3,\beta_4$) and $P_4$ has a root $\beta_7$ ($\neq \beta_1,\beta_2,\beta_3,\beta_5$) such that neither $\beta_6$ nor $\beta_7$ is a root of $P_5$. (Note that $\beta_6$ and $\beta_7$ may or may not be the same.) 

Let $\hat{P}_i$ be the resulting polynomial from $P_i$ by removing $\beta_3,\beta_6,\beta_7$ from $N_i$, and let $\hat{N}_i\triangleq N_i\setminus \{\beta_3,\beta_6,\beta_7\}$. Let $\hat{d}_i\triangleq \deg(\hat{P}_i)$. Note that $(\hat{d}_1,\dots,\hat{d}_5)=(3,3,3,3,4)$. Since the reduction set $\{\beta_3,\beta_6,\beta_7\}$ is acceptable, $W(\hat{P}_1,\dots,\hat{P}_4)\equiv 0$ (by Lemma~\ref{lem:RW}). Thus, $\{\hat{P}_1,\dots,\hat{P}_4\}$ has GRP (by Theorem~\ref{thm:ThirdSp}). This is a contradiction since there exists no $(4,1)$- or $(3,2)$-subset in $\{\hat{P}_1,\dots,\hat{P}_4\}$ as $|\hat{N}_1\cap\hat{N}_3|=|\hat{N}_2\cap\hat{N}_4|=0$, and there exists no $(2,3)$-subset in $\{\hat{P}_1,\dots,\hat{P}_4\}$ as $|\hat{N}_1\cap\hat{N}_2|=2$, $|\hat{N}_1\cap\hat{N}_3|=|\hat{N}_2\cap\hat{N}_4|=0$, and $|\hat{N}_3\cap\hat{N}_4|\leq 2$.
\end{proof}

\bibliographystyle{IEEEtran}
\bibliography{CDERefs}

\end{document}